\documentclass[letterpaper, 10 pt, conference]{ieeeconf}  

\IEEEoverridecommandlockouts                              

\usepackage{graphicx} 
\usepackage{epsfig} 
\usepackage{amsmath} 
\usepackage{amssymb}  
 
\usepackage{amsthm}
\usepackage{todonotes}

\usepackage{xcolor}
\usepackage{tikz}
\usetikzlibrary{positioning,calc}
\usetikzlibrary{arrows.meta}
\usetikzlibrary{shapes}
\usepackage{psfrag}

\usepackage{url}
\usepackage{float}
\usepackage{multirow}
\usepackage{wrapfig}
\usepackage{subcaption}
\usepackage{siunitx}
\usepackage[hidelinks]{hyperref}

\DeclareMathOperator{\sign}{sign}
\DeclareMathOperator{\Imag}{Im}
\DeclareMathOperator{\Real}{Re}
\DeclareMathOperator{\imag}{i}

\newtheorem{theorem}{Theorem}

\newtheorem{remark}{Remark}

\usepackage{tikz}
\usetikzlibrary{shapes,arrows,positioning}
\usepackage{verbatim}

\title{\LARGE \bf
Chattering Reduction for a Second-Order Actuator via Dynamic Sliding Manifolds
}

\author{Patricia N\"other, Lars Watermann and Johann Reger$^\star$
\thanks{All authors are with the Control Engineering Group, Technische Universit\"at Ilmenau, P.O. Box 100565, D-98684, Ilmenau, Germany\newline
        $^\star$Corresponding author: {\tt\small johann.reger@tu-ilmenau.de}}%
}

\begin{document}

\maketitle
\thispagestyle{empty}
\pagestyle{empty}

\begin{abstract}
We analyze actuator chattering in a scalar integrator system subject to second-order actuator dynamics with an unknown time constant and first-order sliding-mode control, using both a conventional static sliding manifold and a dynamic sliding manifold. Using the harmonic balance method, we prove that it is possible to adjust the parameters of the dynamic sliding manifold for the specified system class so as to reduce the amplitude of the chattering in comparison to the static manifold. We illustrate our results with a simulation example. This contribution serves as a proof of concept to motivate further investigations in chattering reduction via dynamic sliding manifolds.

\end{abstract}

\section{Introduction}
Sliding mode control (SMC) is a well-known approach for handling systems with bounded disturbances and uncertainties~\cite{shtesselSlidingModeControl2014}. However, a significant disadvantage is the occurrence of chattering, which is a finite frequency, undesirable periodic oscillation along the sliding manifold~\cite{utkinChatteringProblemSliding2006, levantChatteringAnalysis2010}. 
Chattering is known to occur due to the discontinuity in the control signal. Common methods against this kind of chattering are the approximation of the discontinuity~\cite{burtonContinuousApproximationVariable1986} and higher-order sliding modes by shifting the discontinuity into higher order derivatives of the control~\cite{bartoliniChatteringAvoidanceSecondorder1998}. 
However, in addition, chattering also occurs with continuous signals due to unmodeled dynamics in the system~\cite{boikoAnalysisChatteringContinuous2005, boikoAnalysisChatteringSystems2007}. Especially, chattering is known to occur in first-order SMC with actuators of relative degree two or higher~\cite{boikoAnalysisChatteringContinuous2005}.
To suppress chattering caused by unmodeled dynamics, applying state observers~\cite{leeChatteringSuppressionMethods2007, hsuNyquistCriterionChattering2024} or filter characteristics~\cite{kruppChatteringfreeSlidingMode1999} are discussed. In~\cite{castilloDescribingfunctionbasedAnalysisTune2020} chattering due to discontinuity is reduced by the approximation with a saturation function and the parameters of the control are tuned using the describing function method.
Another approach for adjusting or improving the system behavior is using a dynamic sliding manifold (DSM) instead of a static sliding manifold (SSM). This method is adopted from frequency shaping sliding mode control~\cite{youngFrequencyShapingCompensator1993}.
For example in~\cite{shtesselSlidingModeControl2003}, a DSM is applied to a boost and buck-boost converter to stabilize the internal dynamics.
Also in~\cite{koshkoueiDynamicSlidingMode2005}, the extra dynamics of the DSM are used as a compensator for the desired system behavior.
In our previous study~\cite{tietzeLocalStabilityAnalysis2025}, we analyze stability with both globally bounded and unbounded perturbations using a DSM.
The novel idea now is to use the DSM for reducing chattering compared to a SSM with similar performance properties, but without changing the switching law, i.e.~without smoothening the sign-function, reducing the gain or increasing the order of the sliding mode.
The analysis of the chattering is based on the so-called harmonic balance (HB) method or describing function (DF) technique, rendering possible to calculate the amplitude and frequency of the fundamental oscillation~\cite{slotineAppliedNonlinearControl1991}.
The DF analysis in~\cite{shtesselNewApproachChattering1996} shows that there is no chattering with a first-order stable lag in the unmodeled dynamics. The analysis is performed for ideal and real first-order SMC with a SSM.

In this study, we analyze the capability of a DSM to reduce the chattering amplitude. As a proof of concept we consider a scalar integrator plant with second order actuator dynamics and unknown time constant. We show that there is always a DSM that reduces the chattering amplitude for this specific system. A simulation study underlines the result and demonstrates the effectiveness of the approach.

The manuscript is structured as follows: In~\autoref{sec:problem_statement} we introduce the problem statement, and give the system class and the actuator structure. In~\autoref{sec:SMC} we analyze the chattering of the SSM and calculate its corresponding amplitude. The main results are presented in~\autoref{sec:DSM} with the chattering analysis for the DSM and the proof of reducing the chattering amplitude for appropriate parameter selection in the DSM. These results are illustrated with a simulation example and various actuator values in~\autoref{sec:example}. Finally, we conclude with our results in~\autoref{sec:conclusion}.

\section{Problem Statement}\label{sec:problem_statement}
Consider a scalar system
\begin{align}
\dot x (t) &= u(t) \label{eq:system}
\end{align}
with time $t$, state $x(t)\in \mathbb{R}$, control input $u(t)\in \mathbb{R}$ and initial value $x(t_0) = x_0$ at $t_0 \geq 0$. We control this plant with a first-order SMC referring to the sliding manifold (SM) $\Xi = \left\lbrace x \in \mathbb{R}\, |\, \sigma (x) = 0 \right\rbrace$. The argument of the sliding variable $\sigma$ is omitted for the sake of readability. For the control design we choose the reaching law
\begin{align}
	\dot \sigma  &= -K \sign{\left(\sigma\right)} \label{eq:reaching_law}
\end{align}
with control gain $K >0$ to achieve finite-time convergence to the SM~\cite{shtesselSlidingModeControl2014}. Additionally, assume a second order actuator
\begin{align}
	\begin{pmatrix}
		\dot \xi_{1} (t)\\
		\dot \xi_{2} (t)
	\end{pmatrix} &= \begin{pmatrix}
	0 & 1 \\
	-\frac{1}{\tau^2} & -\frac{2}{\tau}
	\end{pmatrix} \begin{pmatrix}
	\xi_{1} (t)\\
	\xi_{2} (t)
	\end{pmatrix} + \begin{pmatrix}
	0\\
	\frac{1}{\tau^2}
	\end{pmatrix} u_\mathrm{a}(t) \label{eq:actuator_SOL}
\end{align}
with actuator state $\xi(t) = (\xi_1(t),\xi_2(t))^\top \in \mathbb{R}^2$, actuator input $u_\mathrm{a}(t)\in \mathbb{R}$ and initial value $\xi(t_0) = \xi_0$. The actuator output $u(t) = \xi_1(t)$ shall be the input to the plant \eqref{eq:system}. The actual control input is the input of the actuator. We assume the structure \eqref{eq:actuator_SOL} of the actuator to be known, but the time constant $\tau$ to be unknown. Furthermore, we assume that the initial value of the actuator state is not known and that the actuator state is not measurable, hence not available for control. The relative degree of the sliding variable $\sigma$ with respect to the input $u_\mathrm{a}$ is three. Thus, chattering is expected in the closed loop system~\cite{boikoAnalysisChatteringContinuous2005}.
The aim is to reduce the amplitude of the chattering by designing a new DSM with the same relative degree. All solutions of systems with discontinuities are understood in the sense of Filippov~\cite{filippovDifferentialEquationsDiscontinuous1988}.

\section{Chattering Analysis for a Static Sliding Manifold}\label{sec:SMC}
We first design the control law for the ideal case without the actuator. A static sliding manifold (SSM) {$\Xi_{\mathrm{S}}~=~\left\lbrace~x\in\mathbb{R}\,|\, \sigma_{\mathrm{S}}(x)=0\right\rbrace$} with
\begin{align}
	\sigma_{\mathrm{S}} = x(t) \label{eq:SM_static}
\end{align}
and the reaching law \eqref{eq:reaching_law} yields the control law
\begin{align}
	u(t) &= - K \sign (\sigma_\mathrm{S})\label{eq:SMC_control_law}
\end{align}
as input to the plant \eqref{eq:system}. With the actuator, the discontinuous control law \eqref{eq:SMC_control_law} actually acts as the input $u_\mathrm{a}$ to \eqref{eq:actuator_SOL}. To analyze the chattering, we use the HB method~\cite{slotineAppliedNonlinearControl1991}. We follow the steps with unmodeled dynamics from~\cite{shtesselNewApproachChattering1996}.
To this end, we extend the state of the plant \eqref{eq:system} by the actuator states from \eqref{eq:actuator_SOL} and insert the control law \eqref{eq:SMC_control_law}. Replacing the nonlinearity $\sign(\sigma)$ by an auxiliary input $u_\mathrm{HB}$ we get the linear, partially closed loop system
\begin{align}
\Sigma_\mathrm{S}: \left\lbrace 
\begin{aligned}
	\begin{pmatrix}
		\dot x\\
		\dot \xi_1\\
		\dot \xi_2
	\end{pmatrix} &= \begin{pmatrix}
		0 & 1& 0\\
		0 & 0 & 1\\
		0  & -\frac{1}{\tau^2} &-\frac{2}{\tau}
	\end{pmatrix} \begin{pmatrix}
		x\\
		\xi_{1}\\
		\xi_{2}
	\end{pmatrix} + \begin{pmatrix}
		0\\
		0\\
		-\frac{K}{\tau^2}
	\end{pmatrix} u_{\mathrm{HB}} \label{eq:SMC_SOL_partially_closed_loop} \\
    \sigma_\mathrm{S} &=  \begin{pmatrix}
		1 & 0& 0
	\end{pmatrix} \begin{pmatrix}
		x\\
		\xi_{1}\\
		\xi_{2}
	\end{pmatrix}
    \end{aligned}
    \right.
\end{align}
as depicted in~\autoref{fig:partially_closed_loop}.
Note that, in this case, the partially closed loop is similar to the open loop (except for the factor $-K$) because the equivalent control vanishes for the integrator~\eqref{eq:system}.
The corresponding transfer function from auxiliary input $u_\mathrm{HB}$ to output $\sigma_\mathrm{S}$ is given by
\begin{align}
	G_{\mathrm{S}}(s) &= -\frac{K}{s {\left(s^2 \tau^2 +2 s \tau +1\right)}}.\label{eq:SMC_SOL_transfer_function}
\end{align}
Hence, applying the discontinuous input $\sign\left(\sigma \right)$ to the partially closed loop \eqref{eq:SMC_SOL_partially_closed_loop} leads to periodic oscillations in the system. Their fundamental oscillation is
\begin{align}
    \sigma_1(t) &= \hat{\sigma} \sin\left(\omega_\mathrm{p} t\right)
\end{align}
with the amplitude $\hat{\sigma}>0$ and frequency $\omega_\mathrm{p}$~\cite{slotineAppliedNonlinearControl1991}. Assume that there is no phase shift in the oscillation of the sliding variable and that
\begin{align}
    -\sigma &= \hat{\sigma} \sin\left(\omega_\mathrm{p} t\right)
\end{align}
is valid~\cite{shtesselSlidingModeControl2014}. To determine the frequency~$\omega_\mathrm{p}$ and the amplitude $\hat{\sigma}$ of the fundamental oscillation we intersect the negative inverse of the DF~\cite{shtesselSlidingModeControl2014} of the nonlinearity $\sign(\sigma)$, that is
\begin{align}
	N(\hat{\sigma}) &=  \dfrac{4}{\pi \hat{\sigma}} \label{eq:DF_sign}
\end{align}
with \eqref{eq:SMC_SOL_transfer_function} such that
\begin{align}
	-\frac{1}{N(\hat{\sigma})}& =G_{\mathrm{S}}(\imag \omega_{\mathrm{p}}).\label{eq:HB_intersection}
\end{align}
First, we set the imaginary part of \eqref{eq:SMC_SOL_transfer_function} to zero
\begin{align}
	\Imag \lbrace G_{\mathrm{S}}(\imag \omega_{\mathrm{p,S}}) \rbrace &=  \frac{-K \left(\omega_{\mathrm{p,S}}^2 \tau^2 -1\right)}{\omega_{\mathrm{p,S}} \left(4\omega_{\mathrm{p,S}}^2 \tau^2 +{\left(\omega_{\mathrm{p,S}}^2 \tau^2 -1\right)}^2 \right)} = 0
\end{align}
 and receive the frequency
\begin{align}
    \omega_{\mathrm{p,S}} &=\frac{1}{\tau} .\label{eq:SMC_periodic_frequency}
\end{align} Inserting the frequency  \eqref{eq:SMC_periodic_frequency} into the real part of the transfer function \eqref{eq:SMC_SOL_transfer_function}, that is
\begin{align}
	\Real \lbrace G_{\mathrm{S}}(\imag \omega) \rbrace &= \frac{2 K \tau }{4\omega^2 \tau^2 +{{\left(\omega^2 \tau^2 -1\right)}}^2 },
\end{align}
yields
\begin{align}
	\Real \lbrace G_{\mathrm{S}}(\imag \omega_{\mathrm{p,S}}) \rbrace &= \frac{K \tau}{2}.\label{eq:SMC_transfer_function_real_wp}
\end{align}
Then calculating the intersection of the HB \eqref{eq:HB_intersection} with  \eqref{eq:SMC_transfer_function_real_wp} and the DF \eqref{eq:DF_sign} leads to 
\begin{align}
    \hat \sigma_{\mathrm{S}} &= \left| - \frac{2 K \tau} {\pi}\right|\label{eq:SMC_Amp_sigma}
\end{align}
as chattering amplitude of the sliding variable.

\begin{figure}
    \centering
    \tikzstyle{block} = [draw, rectangle, 
    minimum height=3em, minimum width=6em]
\tikzstyle{sum} = [draw, circle, node distance=1cm]
\tikzstyle{input} = [coordinate]
\tikzstyle{output} = [coordinate]
\tikzstyle{pinstyle} = [pin edge={to-,thin,black}]
\tikzset{every arrow/.style={->, >=latex}}

\begin{tikzpicture}[auto, node distance=2cm,>=latex']
	\node [sum, name=input] {};
	\node [block, right of=input, node distance = 2cm,inner sep=5pt] (nonlinearity) {\def\svgwidth{0.2\linewidth}
    \input{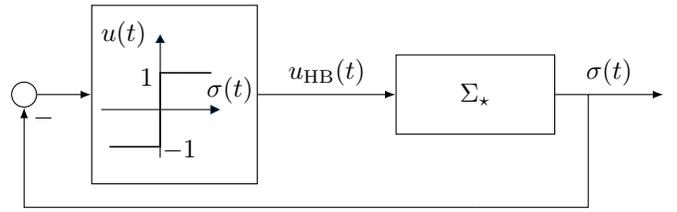}};
	
    \node [block, right of=nonlinearity,
	node distance=4cm, inner sep = 5pt] (linear_part) {$\Sigma_{\star}$};
	\node [output, right of=linear_part, node distance=2.5cm] (output) {};
    \node [output, right of=linear_part, node distance = 1.5cm] (y) {};
	\node[input, name=helpff2, below of = linear_part, node distance = 1.5cm] {};
	
	\draw [every arrow] (input) -- (nonlinearity);
	\draw [every arrow] (nonlinearity) -- node {$u_{\mathrm{HB}}(t)$} (linear_part);
    \draw [every arrow] (linear_part) -- node {$\sigma(t)$} (output);
	\draw [-] (y) |- (helpff2);
	\draw [every arrow] (helpff2) -| node[pos=0.95, right]{$-$} (input);
\end{tikzpicture}
    \vspace{-5ex}
    \caption{Block diagram of the sign-function and the partially closed loop $\Sigma_{\star}$ with $\star \in \left\lbrace \mathrm{S}, \mathrm{D} \right\rbrace$.}
    \label{fig:partially_closed_loop}
\end{figure}

\section{Dynamic Sliding Manifold}\label{sec:DSM}
We may adhere to the same procedure for analyzing the chattering calculation as before in \autoref{sec:SMC}. Let the DSM $\Xi_{\mathrm{D}} = \left\lbrace x \in  \mathbb{R},\,z \in \mathbb{R}\, |\, \sigma_{\mathrm{D}} (x,z) = 0 \right\rbrace$ be given by
\begin{subequations}
	\begin{align}
		\dot z(t) &= f z(t) + g x(t)\\
		\sigma_{\mathrm{D}} &= h z(t) + l x(t) \label{eq:SM_dynamic}
	\end{align} \label{eq:DSM}%
\end{subequations}
from~\cite{youngFrequencyShapingCompensator1993} with actuator state $z(t)\in\mathbb{R}$ and design parameters $f,g,h,l \in \mathbb{R}$. We assume $l \neq 0$ so as to keep the relative degree one from $x$ to $\sigma$ as in \autoref{sec:SMC}. With the reaching law \eqref{eq:reaching_law} we obtain the control law
\begin{align}
	u(t) &= -l^{-1}\left(K \sign(\sigma_\mathrm{D}) + h f z(t) + h g x(t)\right)\label{eq:DSM_control_law}
\end{align} 
including the equivalent control.

\subsection{Stability Constraints}
The DSM \eqref{eq:DSM} itself must be asymptotically stable, so we need $f < 0$.
Since the plant \eqref{eq:system} does not contain perturbations,  \eqref{eq:DSM_control_law} with $K >0$ guarantees the existence of the sliding mode~\cite{shtesselSlidingModeControl2014, youngFrequencyShapingCompensator1993}. In sliding phase $0 \equiv \sigma$, the reduced system with $x(t) = - l^{-1} h  z(t)$ reads
\begin{align}
	\dot z (t) &= (f-gl^{-1}h) z(t)
\end{align}
and should be asymptotically stable without the influence of the actuator, hence
\begin{align}
    f-gl^{-1}h <0 \label{eq:DSM_stability_reduced}
\end{align} 
is required for the design of the DSM.

\subsection{Chattering Analysis}
For the analysis with the DSM, we calculate the partially closed loop equivalently to \autoref{sec:SMC}. Therefore, we extend the state $x$ of \eqref{eq:system} by the additional state $z$ of the sliding manifold \eqref{eq:DSM} and the actuator state $\xi $ from \eqref{eq:actuator_SOL}. Inserting the control law \eqref{eq:DSM_control_law} and again replacing $\sign (\sigma)$ by the auxiliary input $u_{\mathrm{HB}}$ we arrive at the linear, partially closed loop
\begin{align}
\Sigma_\mathrm{D}: \left\lbrace 
\begin{aligned}
	\begin{pmatrix}
		\dot z\\
		\dot x\\
		\dot \xi_1\\
		\dot \xi_2
	\end{pmatrix} &= \begin{pmatrix}
		f & g & 0 & 0\\
		0 & 0 & 1& 0\\
		0 & 0 & 0 & 1\\
		-\frac{h f}{l\tau^2} & -\frac{h g}{l\tau^2}  & -\frac{1}{\tau^2} & -\frac{2}{\tau}
	\end{pmatrix} \begin{pmatrix}
		z\\
		x\\
		\xi_{1}\\
		\xi_{2}
	\end{pmatrix}\\
    &\quad+ \begin{pmatrix}
		0\\
		0\\
		0\\
		-\frac{K}{l\tau^2}
	\end{pmatrix} u_{\mathrm{HB}}\\
    \sigma_{\mathrm{D}} &= \begin{pmatrix}
		h & l & 0 & 0\\
	\end{pmatrix} \begin{pmatrix}
		z\\
		x\\
		\xi_{1}\\
		\xi_{2}
	\end{pmatrix}.\label{eq:DSM_SOL_partially_closed_loop}
    \end{aligned}\right.
\end{align} 
Then the transfer function from input $u_\mathrm{HB}$ to output $\sigma_\mathrm{D}$ is
\begin{align}
        G_{\mathrm{D},\sigma}(s)= \frac{-K \left(g h-f l+ls\right)}{s \left(l \tau^2  s^3 +a_2  s^2  +a_1 s + g h - f l\right)}\label{eq:DSM_SOL_transfer function}
\end{align}
with parameters $a_2 = (2l\tau -f l \tau^2)$, $a_1 = l -2 f l\tau$. The parameters $f,g,h,l$  have to be selected such that $l \tau^2  s^3 +a_2  s^2 \tau +a_1 s + g h - f l$ is a Hurwitz polynomial. This is especially fulfilled if $f<0$, $\sign(h) = -\sign(l)$ and \eqref{eq:DSM_stability_reduced} holds.
\begin{theorem}\label{lemma:mag_sigma}
	For system \eqref{eq:system} subject to actuator dynamics \eqref{eq:actuator_SOL} with $\tau$ unknown,
    a dynamic sliding manifold \eqref{eq:DSM} can be chosen such that the  chattering amplitude of the sliding variable $\sigma$ is smaller than for the static sliding manifold~\eqref{eq:SM_static}.
\end{theorem}
\begin{proof}
	We choose scaled variables $\tilde{s}:= s \tau$, $\tilde{f} := f \tau$ and $\tilde{g} := g \tau$ according to~\cite{filippovDifferentialEquationsDiscontinuous1988}.
    The scaling of the parameter does not influence the stability constraints
    \begin{align}
        \tilde{f}-\tilde{g}l^{-1}h =\left(f -  g h l^{-1}\right)\tau <0 \Leftrightarrow  f-gl^{-1}h <0.
     \end{align}
     Additionally, we set $h=-1$ and $l=1$. With these substitutions, the transfer function~\eqref{eq:DSM_SOL_transfer function} simplifies to 
    \begin{align}
	G_{\mathrm{D},\sigma}(\tilde{s}) &= \frac{-K \tau \left(\tilde{f}+\tilde{g}-\tilde{s}\right)}{\tilde{s} \left(-\tilde{s}^3 +(\tilde{f}-2) \tilde{s}^2 +(2\tilde{f}-1) \tilde{s} + \tilde{f}+\tilde{g}   \right)}.\label{eq:DSM_transfer_function_scaled_subs}
	\end{align}
    Then we analyze the limit $\tilde{g} \rightarrow -\tilde{f}$. Condition $|\tilde{g}| < |\tilde{f}|$ must hold to satisfy the stability constraint~\eqref{eq:DSM_stability_reduced}. In the limit, the transfer function~\eqref{eq:DSM_transfer_function_scaled_subs} is
	\begin{align}
		\lim\limits_{\tilde{g} \rightarrow -\tilde{f}} G_{\mathrm{D},\sigma}(\tilde{s}) &= -\frac{K \tau }{\tilde{s} {\left(+\tilde{s}^2 + (2-\tilde{f}) \tilde{s}-2 \tilde{f} +1\right)}} \label{eq:DSM_SOL_transfer_function_lim}.
	\end{align}
    Note that this limit cannot be reached. Yet, $\tilde{g}$ can be chosen arbitrarily close to $-\tilde{f}$, hence, it is possible to get arbitrarily close to the limit. Calculating the intersection of the imaginary part of~\eqref{eq:DSM_SOL_transfer_function_lim} with the real axis yields
	\begin{align}
	0 &= \lim\limits_{\tilde{g} \rightarrow -\tilde{f}} \Imag \lbrace G_{\mathrm{D},\sigma}(\imag \tilde{\omega}_{\mathrm{p}}) \rbrace \nonumber\\
    &= -\frac{K \tau \left(\tilde{\omega}_{\mathrm{p}}^2 +2 \tilde{f}-1\right)}{\tilde{\omega}_{\mathrm{p}} \left(\left(\tilde{\omega}_{\mathrm{p}}^2 +2 \tilde{f}-1\right)^2 +\left(2 \tilde{\omega}_{\mathrm{p}} -\tilde{f} \tilde{\omega}_{\mathrm{p}} \right)^2 \right)}\\
	&\Rightarrow \lim\limits_{\tilde{g} \rightarrow -\tilde{f}}  \tilde{\omega}_{\mathrm{p,D}} = \sqrt{1-2 \tilde{f}}.\label{eq:DSM_periodic_frequency}
	\end{align}
 Inserting~\eqref{eq:DSM_periodic_frequency} into the real part of the limit~\eqref{eq:DSM_SOL_transfer_function_lim} gives
 \begin{align}
 	&\lim\limits_{\tilde{g} \rightarrow -\tilde{f}} \Real \lbrace G_{\mathrm{D,\sigma}}\left(\imag \tilde{\omega}_{\mathrm{p,D}} \right) \rbrace\nonumber\\
    &= \frac{K \tau \left(2 -\tilde{f} \right)}{ {{\left(\tilde{\omega}_{\mathrm{p,D}}^2 +2 \tilde{f}-1\right)}}^2 +{{\left(2 \tilde{\omega}_{\mathrm{p,D}} -\tilde{f} \tilde{\omega}_{\mathrm{p,D}} \right)}}^2 }\\
 	&= \frac{K \tau }{2 \tilde{f}^2 -5 \tilde{f}+2}
 \end{align}
and equating with the negative inverse of the DF~\eqref{eq:DF_sign} yields the chattering amplitude  
\begin{align}
	\lim\limits_{\tilde{g} \rightarrow -\tilde{f}} \hat{\sigma}_{\mathrm{D}} = \left| -\frac{4 K \tau }{\pi {\left(2 \tilde{f}^2 -5 \tilde{f}+2\right)}}\right|\label{eq:DSM_magnitude_sigma}
\end{align}
of the sliding variable. Eventually, observe that
\begin{align}
    &&\lim\limits_{\tilde{g} \rightarrow -\tilde{f}} \hat{\sigma}_{\mathrm{D}} &< \hat{\sigma}_{\mathrm{S}}\\
	&\Leftrightarrow&\frac{4 K \tau }{\pi {\left(2 \tilde{f}^2 -5 \tilde{f}+2\right)}} &< \frac{2 K \tau}{\pi}\\
	&\Leftrightarrow& 2 \tilde{f}^2 -5 \tilde{f} &> 0,
\end{align}
which is fulfilled for all $\tilde{f} < 0$, mimicking also one of the stability constraints. 
\end{proof}

\begin{remark}
    The parameters $\tilde{f}$ and $\tilde{g}$ are unknown in general, but are scaled in the same way. Hence, the limit $\tilde{g} \rightarrow -\tilde{f}$ is equivalent to the limit $g \rightarrow -f$ with $f<0$. With the selection $l = 1$ and $h =-1$ all parameters can be chosen directly. To overcome the inaccuracies of the HB method, $-f$ should be selected sufficiently large.
\end{remark}
\begin{remark}
	The selection of the parameters in the proof of \autoref{lemma:mag_sigma} is only one possibility. To avoid the limit, it is also feasible to set $g =-\alpha \tilde{f}, 0<\alpha<1$ and choose $\tilde{f}$.
    This approach leads to a more sophisticated analysis including a graphical comparison of the amplitudes.
\end{remark}
Since the sliding variables $\sigma_\mathrm{S}$ and $\sigma_\mathrm{D}$ correspond to distinct manifolds, the reduction of the chattering amplitude of the sliding variable does not guarantee the reduction of the chattering amplitude of the state in general, i.e., $\hat{\sigma}_\mathrm{D} <\hat{\sigma}_\mathrm{S}$ does not ensure $\hat{x}_\mathrm{D} <\hat{x}_\mathrm{S}$. It is therefore essential to take the chattering of the state into account.

\begin{theorem}\label{theorem:magnitude_x}
	For system~\eqref{eq:system} subject to actuator dynamics~\eqref{eq:actuator_SOL} with $\tau$ unknown,
    a dynamic sliding manifold~\eqref{eq:DSM} can be chosen such that the amplitude of the chattering of the state $x(t)$ is smaller than for the static sliding manifold~\eqref{eq:SM_static}.
\end{theorem}
\begin{proof}
The amplitude of $x$ with the SSM
\begin{align}
   \hat{x}_{\mathrm{S}} = \hat{\sigma}_{\mathrm{S}} = \frac{2 K\tau}{\pi}\label{eq:SMC_magnitude_x}
\end{align}
is the same as~\eqref{eq:SMC_Amp_sigma}. Due to the extended dynamics of~\eqref{eq:SM_dynamic}, the amplitude of $x$ is different from that of $\sigma_{\mathrm{d}}$ in the case with a DSM.
Since the partially closed loop~\eqref{eq:DSM_SOL_partially_closed_loop} is linear, 
the fundamental oscillation appears in all signals with the same frequency $\omega_\mathrm{p}$. The amplitude and phase of the oscillation can be calculated by means of frequency domain analysis.
 To calculate the amplitude of $u_{\mathrm{HB}}$ we consider
\begin{align}
    \hat{u}_{\mathrm{HB}} &= \frac{\hat{\sigma}_{\mathrm{D}}}{| G_{\mathrm{D},\sigma}(\imag \omega_{\mathrm{p}}) |}
    =\frac{\hat{\sigma}_{\mathrm{D}}}{|\Real \lbrace G_{\mathrm{D},\sigma}(\imag \omega_{\mathrm{p}}) \rbrace|} = \frac{4}{\pi}.\label{eq:magnitude_uHB}
\end{align}
We only need the absolute value of the real part because the imaginary part here is zero at $\omega_{\mathrm{p}}$. 
Calculating the transfer function from auxiliary input $u_{\mathrm{HB}}(t)$ to the state $x(t)$ in system~\eqref{eq:DSM_SOL_partially_closed_loop} yields
\begin{align}
        G_{\mathrm{D},x}(s)= \frac{K \left(f -s\right)}{s \left(l \tau^2  s^3 +a_2  s^2  +a_1 s +gh - fl \right)}\label{eq:DSM_SOL_transferfunction_x}
\end{align}
with $a_2 = (2l\tau-fl\tau^2)$ and $a_1 = (l-2fl\tau)$. We choose the same scaling of $\tilde{s}$, $\tilde{f}$ and $\tilde{g}$, as well  as the same selection $h = -l = -1$ as in the proof of \autoref{lemma:mag_sigma}. Considering the limit for $\tilde{g} \rightarrow -\tilde{f}$ again yields
\begin{align}
	\lim\limits_{\tilde{g} \rightarrow -\tilde{f}} G_{\mathrm{D,x}}(\tilde{s}) &=  \frac{K \tau\left(\tilde{f}-\tilde{s}\right)}{\tilde{s}^2 {\left(\tilde{s}^2 + (2-\tilde{f}) \tilde{s}-2 \tilde{f} +1\right)}} \label{eq:DSM_SOL_transfer_function_x}
\end{align}
as limit of the new transfer function. Then,  with the amplification of the transfer function~\eqref{eq:DSM_SOL_transfer_function_x} we calculate the amplitude
\begin{align}
    \lim\limits_{\tilde{g} \rightarrow -\tilde{f}} \hat{x}_{\mathrm{D}} &= |G_{\mathrm{D},x}(\imag \omega_{\mathrm{p,D}})|\,\hat{u}_{\mathrm{HB}}\\
    &= \left|\frac{4 K \tau  \left(\tilde{f}-1\right)}{\pi \left(1-2 \tilde{f}\right)^{1/2} \left(2\tilde{f}^2-5\tilde{f}+2\right)}\right|.\label{eq:DSM_magnitude_x}
\end{align}

Eventually, we have to show that
\begin{gather}
\lim\limits_{\tilde{g} \rightarrow -\tilde{f}} \hat{x}_{\mathrm{D}} < \hat{x}_{\mathrm{S}}\\ 
\Leftrightarrow\quad\left|\frac{2  \left(\tilde{f}-1\right)}{ \left(1-2 \tilde{f}\right)^{1/2} \left(2\tilde{f}^2-5\tilde{f}+2\right)}\right| < 1.
\end{gather}
Since $\tilde{f} <0$, $\left(1-2\tilde{f}\right)^{\frac{1}{2}}>1$ also 
\begin{align}
\left(2\tilde{f}^2-5\tilde{f}+2\right) &> 2 (1-\tilde{f})\label{eqn:lars_inequality}
\end{align}
holds. Hence, $\hat{x}_{\mathrm{D}} < \hat{x}_{\mathrm{S}}$ holds for $f<0$ with $h = -1$, $l=1$ and $g \rightarrow -f$.
\end{proof}

Since the selection of the parameters in \autoref{lemma:mag_sigma} and \autoref{theorem:magnitude_x} is the same, we have shown that it is possible to select $f$, $g$, $h$ and $l$ such that the amplitude of the chattering of both $\hat{\sigma}$ and $\hat{x}$ for system \eqref{eq:system} with the DSM is lower than for the SSM. 

\section{Example}\label{sec:example}
We demonstrate the proposed method on two examples with 
$\tau \in \lbrace 0.01\, , 0.1 \rbrace$ and $K=1$. As in the proof of \autoref{lemma:mag_sigma} we choose $h= -1$ and $l = 1$. Since with limit $g \to -f$ the asymptotic stability of the reduced order dynamics is no longer valid in the sliding phase we set $f = -40$ and $g = 0.98 f = 39.6$. Like in integral SMC~\cite{utkinIntegralSlidingMode1996a} we initialize the state $z$ on the sliding manifold $\sigma_\mathrm{D}(t_0) = 0$, i.e. $z_0 = x_0$. Hence, the reaching phase is eliminated.

\begin{figure}
    \centering
    \includegraphics[width=\linewidth]{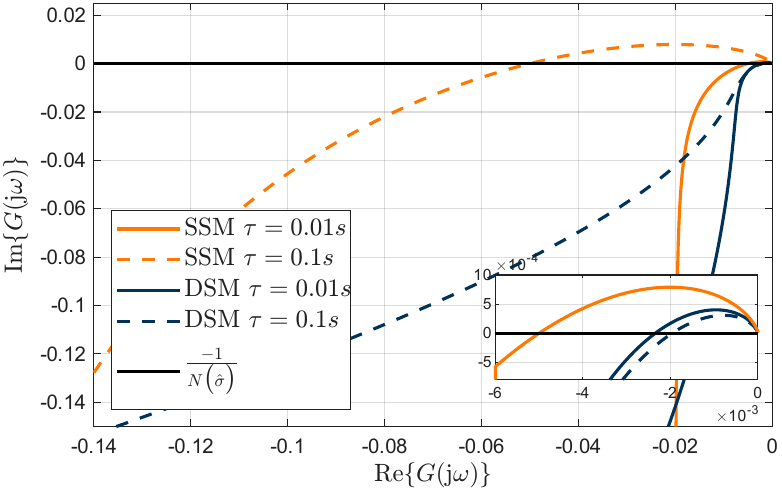}
    \caption{Harmonic Balance analysis for SMC and DSM each with $\tau = \SI{0.01}{s}$ and $\tau = \SI{0.1}{s}$}
    \label{fig:harmonic_balance}
\end{figure}
The HB analysis, \autoref{fig:harmonic_balance}, shows that the expected chattering amplitude of $\sigma$ is significantly lower with the DSM than for the SSM. The amplitude increases with $\tau$ for the SSM. 
For ideal switching, i.e.~with the $\sign$-function, the amplitude of the DSM is lower for larger values of $\tau$. Further, there is a margin between the real axis and the intersection between the local curves of the SSM and DSM. Hence, the approach may also be useful for chattering reduction in systems with nonideal switching, e.g.~with a relay-function, whose negative inverse of the describing function has a negative imaginary part.

\begin{figure}
    \centering
    \includegraphics[width=\linewidth]{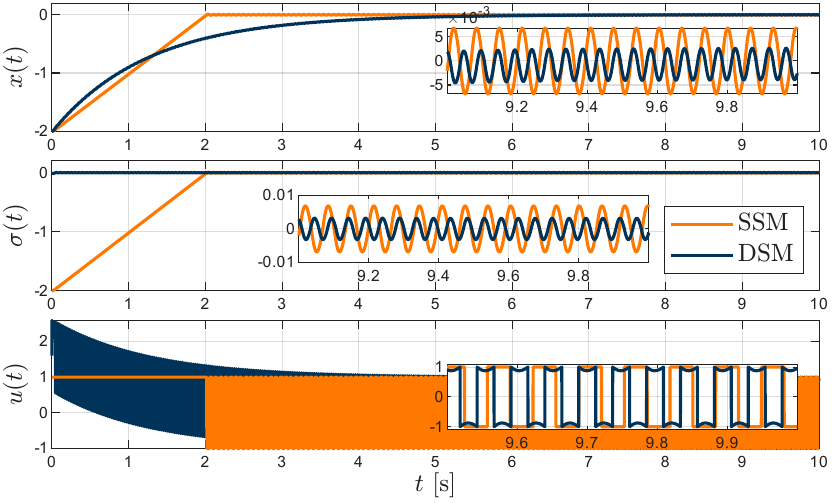}
    \caption{Simulation with $\tau = \SI{0.01}{s}$}
    \label{fig:Simulation_tau_001}
\end{figure}

\begin{figure}
    \centering
    \includegraphics[width=\linewidth]{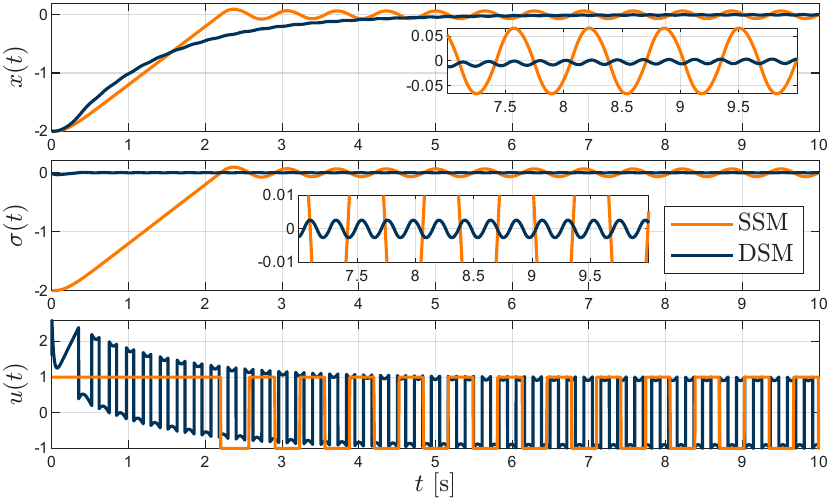}
    \caption{Simulation with $\tau = \SI{0.1}{s}$}
    \label{fig:Simulation_tau_010}
\end{figure}

In \autoref{fig:Simulation_tau_001} and \autoref{fig:Simulation_tau_010} the simulation with different $\tau$ is shown. The amplitude and frequency of the chattering depend on the actuator.  With the selected parameters of the DSM, chattering both with respect to $\sigma$ and $x$ is significantly lower with the DSM than for the SSM. However, the frequency in each scenario is higher with the DSM, which may cause problems in mechanical or electrical systems. If a possible interval for $\tau$ is known, $f$ can be chosen considering the maximal acceptable frequency~\eqref{eq:DSM_periodic_frequency} of the interval for $\tau$. 

\begin{table}
    \centering
    \caption{Comparison of the chattering amplitudes in the analysis with HB and the simulation}
    \begin{tabular}{|c|c|c|c|c|}
    \hline
        \textbf{Scenario} & \textbf{$\hat{\sigma}$ HB} & \textbf{$\hat{\sigma}$ Sim.}  & \textbf{$\hat{x}$ HB} & \textbf{$\hat{x}$ Sim.}\\
        \hline
        \hline
        SSM, $\tau = \SI{0.01}{s}$ & $0.0064$ & $0.0068$ & $0.0064$ & $0.0068$\\
        \hline
        SSM, $\tau = \SI{0.1}{s}$ & $0.0637$ & $0.0660$ & $0.0637$ & $0.0660$\\
        \hline
        DSM, $\tau = \SI{0.01}{s}$ & $0.0029$ & $0.0029$ & $0.0031$ & $0.0032$\\
        \hline
        DSM, $\tau = \SI{0.1}{s}$ & $0.0024$ & $0.0025$ & $0.0039$ & $0.0035$\\
        \hline
    \end{tabular}
    \label{tab:compare_magnitudes}
\end{table}

Table \ref{tab:compare_magnitudes} compares the amplitudes of the simulation and calculation with $\hat \sigma $ from~\eqref{eq:SMC_Amp_sigma} and ~\eqref{eq:DSM_magnitude_sigma} and $\hat x$ from~\eqref{eq:SMC_magnitude_x} and~\eqref{eq:DSM_magnitude_x}, respectively. In this example, the calculated values for the SSM are slightly lower than those from the simulation.  In this particular instance, the precision of the calculation with the HB is sufficiently high to obtain nearly identical results when comparing the scenario with the SSM and the DSM, both with the HB and the simulation. The magnitude of the HB analysis is consistent with that of the simulation.  

\section{Conclusion}\label{sec:conclusion}
We analytically proved that it is always possible, for the considered integrator with second order actuator dynamics, to select DSM parameters such that the chattering amplitude of both the sliding variable and the state is reduced compared to the SSM. The example also illustrates that the reduction of chattering is indeed significant.
Our results constitute a proof of concept and show that the method is a promising technique for chattering reduction. In order to prepare the approach for practical applications, future research will address more general systems and also take the control effort into account.









\bibliography{2ChatFreeSMC.bib}

\end{document}